\documentclass[pra,superscriptaddress,nopacs,twocolumn]{revtex4-2}
 \usepackage{algorithm,algorithmic}

\usepackage[shortlabels]{enumitem}

\usepackage{graphicx,epic,eepic,epsfig,amsmath,latexsym,amssymb,verbatim,color}
\usepackage{amsthm}
\usepackage{thmtools, thm-restate}
\usepackage{mathtools}
\usepackage{amsfonts}       
\usepackage{nicefrac}       
\usepackage{bbm}
\usepackage{float}
\usepackage{tikz}
\usetikzlibrary{chains}
\usetikzlibrary{fit}
\usepackage{pgflibraryarrows}		
\usepackage{epsfig}
\usetikzlibrary{shapes.symbols,patterns} 
\usepackage{pgfplots}

\usepackage[strict]{changepage}
\usepackage{hyperref}
\hypersetup{colorlinks=true,citecolor=blue,linkcolor=blue,filecolor=blue,urlcolor=blue,breaklinks=true}

\usepackage[marginal]{footmisc}
\usepackage{url}

\newtheorem{proposition}{Proposition}

\newtheorem{lemma}[proposition]{Lemma}

\newtheorem{corollary}[proposition]{Corollary}

\newcounter{example}

\mathchardef\ordinarycolon\mathcode`\:
\mathcode`\:=\string"8000
\def\vcentcolon{\mathrel{\mathop\ordinarycolon}}
\begingroup \catcode`\:=\active
\lowercase{\endgroup
	\let :\vcentcolon
}

\usepackage{cleveref}
\usepackage{graphicx}

\graphicspath{{images/}}
\usepackage{xcolor}

\newcommand{\nc}{\newcommand}
\nc{\rnc}{\renewcommand}

\nc{\lbar}[1]{\overline{#1}}
\nc{\bra}[1]{\langle#1|}
\nc{\ket}[1]{|#1\rangle}
\nc{\ketbra}[2]{|#1\rangle\!\langle#2|}
\nc{\braket}[2]{\langle#1|#2\rangle}

\DeclarePairedDelimiter\abs{\lvert}{\rvert}%
\DeclarePairedDelimiter\norm{\lVert}{\rVert}%

\makeatletter
\let\oldabs\abs
\def\abs{\@ifstar{\oldabs}{\oldabs*}}
\let\oldnorm\norm
\def\norm{\@ifstar{\oldnorm}{\oldnorm*}}
\makeatother

\nc{\proj}[1]{| #1\rangle\!\langle #1 |}
\nc{\avg}[1]{\langle#1\rangle}
\nc{\rank}{\operatorname{Rank}}
\nc{\smfrac}[2]{\mbox{$\frac{#1}{#2}$}}
\nc{\tr}{\operatorname{Tr}}
\nc{\ox}{\otimes}
\nc{\dg}{\dagger}
\nc{\dn}{\downarrow}
\nc{\cA}{{\cal A}}
\nc{\cB}{{\cal B}}
\nc{\cC}{{\cal C}}
\nc{\cD}{{\cal D}}
\nc{\cE}{{\cal E}}
\nc{\cF}{{\cal F}}
\nc{\cG}{{\cal G}}
\nc{\cH}{{\cal H}}
\nc{\cI}{{\cal I}}
\nc{\cJ}{{\cal J}}
\nc{\cK}{{\cal K}}
\nc{\cL}{{\cal L}}
\nc{\cM}{{\cal M}}
\nc{\cN}{{\cal N}}
\nc{\cO}{{\cal O}}
\nc{\cP}{{\cal P}}
\nc{\cQ}{{\cal Q}}
\nc{\cR}{{\cal R}}
\nc{\cS}{{\cal S}}
\nc{\cT}{{\cal T}}
\nc{\cV}{{\cal V}}
\nc{\cX}{{\cal X}}
\nc{\cY}{{\cal Y}}
\nc{\cZ}{{\cal Z}}
\nc{\cW}{{\cal W}}
\nc{\csupp}{{\operatorname{csupp}}}
\nc{\qsupp}{{\operatorname{qsupp}}}
\nc{\var}{{\operatorname{var}}}
\nc{\rar}{\rightarrow}
\nc{\lrar}{\longrightarrow}
\nc{\polylog}{{\operatorname{polylog}}}
\nc{\wt}{{\operatorname{wt}}}
\nc{\supp}{{\operatorname{supp}}}

\nc{\argmin}{{\operatorname{argmin}}}

\def\a{\alpha}

\def\t{\theta}

\nc{\RR}{{{\mathbb R}}}
\nc{\CC}{{{\mathbb C}}}
\nc{\FF}{{{\mathbb F}}}
\nc{\NN}{{{\mathbb N}}}
\nc{\ZZ}{{{\mathbb Z}}}
\nc{\PP}{{{\mathbb P}}}
\nc{\QQ}{{{\mathbb Q}}}
\nc{\UU}{{{\mathbb U}}}
\nc{\EE}{{{\mathbb E}}}
\nc{\id}{{\operatorname{id}}}

\nc{\CHSH}{{\operatorname{CHSH}}}

\nc{\<}{\langle}
\rnc{\>}{\rangle}
\nc{\rU}{\mbox{U}}

\nc{\ob}[1]{#1}

\nc{\SEP}{{\text{\rm SEP}}}
\nc{\NS}{{\text{\rm NS}}}
\nc{\LOCC}{{\text{\rm LOCC}}}
\nc{\PPT}{{\text{\rm PPT}}}
\nc{\EXT}{{\text{\rm EXT}}}
\nc{\Sym}{{\operatorname{Sym}}}

\nc{\ERLO}{{E_{\text{r,LO}}}}
\nc{\ERLOCC}{{E_{\text{r,LOCC}}}}
\nc{\ERPPT}{{E_{\text{r,PPT}}}}
\nc{\ERLOCCinfty}{{E^{\infty}_{\text{r,LOCC}}}}
\nc{\Aram}{{\operatorname{\sf A}}}

\usepackage{tikz}
\usepackage{hyperref}
\hypersetup{colorlinks=true,citecolor=blue,linkcolor=blue,filecolor=blue,urlcolor=blue,breaklinks=true}
\usepackage{listings}
\usepackage{caption}
\usepackage{subcaption}
\usepackage[utf8]{inputenc}
\usepackage[T1]{fontenc}

\usepackage{bm}
\usepackage{times}
\definecolor{beamer}{rgb}{0.2,0.2,0.7}
\definecolor{colorone}{rgb}{1,0.36,0.03}
\definecolor{colortwo}{rgb}{0.4,0.77,0.17}
\definecolor{colorthree}{rgb}{0.01,0.51,0.93}
\definecolor{colorfour}{rgb}{0.47,0.26,0.58}
\definecolor{colorfive}{rgb}{0.12,0.55,0.16}
\usepackage{tcolorbox}
\usepackage{relsize}
\usepackage{graphicx}
\usepackage[qm]{qcircuit}
\nc{\st}{\text{subject to} \ }
\nc{\supre}{\text{supremum} \ }
\nc{\sdp}{\text{sdp}}
\usepackage{array}
\usepackage{multirow}
\usepackage{booktabs}
\usepackage{tabularx}

\newcommand{\update}[1]{\textcolor{black}{#1}}
\newcommand{\reupdate}[1]{\textcolor{black}{#1}}
\allowdisplaybreaks

\begin{document}
\title{Optimal quantum dataset for learning a unitary transformation}

\author{Zhan Yu}
\affiliation{Institute for Quantum Computing, Baidu Research, Beijing 100193, China}

\author{Xuanqiang Zhao}
\affiliation{Institute for Quantum Computing, Baidu Research, Beijing 100193, China}

\author{Benchi Zhao}
\affiliation{Institute for Quantum Computing, Baidu Research, Beijing 100193, China}

\author{Xin Wang} \email{wangxin73@baidu.com}
\affiliation{Institute for Quantum Computing, Baidu Research, Beijing 100193, China}

\begin{abstract}
Unitary transformations formulate the time evolution of quantum states. How to learn a unitary transformation efficiently is a fundamental problem in quantum machine learning. The most natural and leading strategy is to train a quantum machine learning model based on a quantum dataset. Although the presence of more training data results in better models, using too much data reduces the efficiency of training. In this work, we solve the problem on the minimum size of sufficient quantum datasets for learning a unitary transformation exactly, which reveals the power and limitation of quantum data. First, we prove that the minimum size of a dataset with pure states is $2^n$ for learning an $n$-qubit unitary transformation. To fully explore the capability of quantum data, we introduce a practical quantum dataset consisting of $n+1$ elementary tensor product states that are sufficient for exact training. The main idea is to simplify the structure utilizing decoupling, which leads to an exponential improvement in the size of the datasets with pure states. Furthermore, we show that the size of the quantum dataset with mixed states can be reduced to a constant, which yields an optimal quantum dataset for learning a unitary. We showcase the applications of our results in oracle compiling and Hamiltonian simulation. Notably, to accurately simulate a 3-qubit one-dimensional nearest-neighbor Heisenberg model, our circuit only uses $96$ elementary quantum gates, which is significantly less than $4080$ gates in the circuit constructed by the Trotter-Suzuki product formula.
\end{abstract}

\date{\today}
\maketitle

\section{Introduction}
Machine learning is a task that builds a model to learn an unknown function based on a training dataset. The training dataset is a set of example input-output pairs repeatedly used during the learning process and is used to fit the parameters of the model. Machine learning has been used in a wide variety of applications~\cite{LeCun2015a}, such as computer vision, natural language processing, and speech recognition. 
At the same time, quantum computing, a technology that harnesses the laws of quantum mechanics to solve problems too complicated for classical computers, has been rapidly advancing. 
Inspired by the powerful capacity of machine learning~\cite{LeCun2015a}, it is natural to develop their quantum counterparts and try to gain more benefits,
which has given rise to an emerging research area, i.e., \emph{quantum machine learning} (QML)~\cite{Biamonte2017a,Schuld2018b,Arunachalam2017}.

QML is built on two components: models and data, both of which could be in the quantum version. QML models could be quantum analogs of classical machine learning models, such as quantum neural networks~\cite{Farhi2018,Cong2018,Benedetti2019a,Yu2022a}, quantum autoencoders~\cite{Romero2017,Wan2017,Cao2020}, and quantum kernel methods~\cite{Schuld2018c,Huang2020a,Liu2020}. Quantum data are quantum states that are generated by quantum processes. Quantum data could be sampled from a natural process like a chemical reaction, or an artificial quantum system, e.g., a quantum computer.

In quantum mechanics, the time evolution of a quantum state according to the Schr{\"o}dinger equation is mathematically represented by a unitary operator. A fundamental problem in quantum computing, \emph{unitary learning}, consists in reproducing an unknown unitary transformation. Unitary learning has been studied intensively in classical approaches~\cite{Arjovsky2016, Hyland2017}, quantum settings~\cite{Bisio2010, Marvian2016}, and hybrid quantum-classical schemes~\cite{Heya2018, Khatri2019, Sharma2019, Jones2022}. One of the most natural and leading methods is to train a QML model based on a quantum dataset~\cite{Cincio2018, Beer2020, Cincio2021, Caro2021}. It raises a fundamental question: \emph{what is the minimum size of the quantum training dataset that is sufficient for learning a unitary?}

The power and limitation for learning an unknown unitary based on datasets with pure states are investigated in Ref.~\cite{Poland2020}, which suggests that the size of the dataset grows exponentially with the system size. The exponential scaling on the size of training dataset leads to an exponential overhead of the training process, which limits the efficiency of QML. With a sufficient ancillary system, it is not difficult to see that one Choi-Jamio\l{}kowski state is sufficient for learning the unitary~\cite{Sharma2020a, Chakrabarti2019, Cincio2021} by the Choi-Jamio\l{}kowski isomorphism~\cite{Choi1975,Jamiokowski1972}. However, the costs of using the ancillary quantum system and computation in a space with doubled dimensions are also expensive. Can we reduce and find the optimal size of the training dataset for an ancilla-free setting?

In this work, we resolve this problem by establishing optimal quantum datasets for learning a unitary transformation. We first introduce a formal definition of the minimum training dataset problem for unitary learning. We analyze the minimum size of training datasets with pure states, and show that using $2^n$ nonorthogonal but linearly independent pure states as the training dataset is optimal for learning an $n$-qubit unitary. In order to exploit the power of the quantum dataset and reduce the exponential size, we introduce datasets with mixed states for learning a unitary transformation. In particular, we leverage the idea of decoupling to construct an efficient and practical dataset with $n+1$ elementary tensor product states that is sufficient to learn an $n$-qubit unitary operator. This implies that quantum datasets with certain structures could be extremely efficient for learning an arbitrary unitary transformation. We further reduce the size of the dataset to two and prove that this is optimal for learning a unitary with ancilla-free systems. More generally, we prove that a quantum dataset consisting of two randomly generated mixed states is sufficient for learning an unknown unitary. 

We showcase practical applications of our results under the framework of hybrid quantum-classical algorithms. We apply our results to do Hamiltonian simulations and show that unitary learning can significantly reduce the depth of quantum circuits, compared to the traditional method using product formulas~\cite{suzuki1991general}. Another application is oracle compiling, which helps in transforming a high-level quantum algorithm (e.g., Grover's algorithm~\cite{Grover1996}) into a sequence of elementary quantum gates that could run on near-term quantum devices.

Our results on the optimal quantum dataset for unitary learning notably reduce the computational cost of training processes, from an exponential overhead to a constant. On the other hand, since the quantum data are generated by sampling from quantum processes, the optimal size of the dataset also leads to the optimal cost of data sampling. Our results also bring new insights regarding the power of quantum data. Quantum data used in QML can be considerably different than their classical counterparts, since quantum mechanics permits quantum states to be mixed, entangled, measured, distilled, concentrated, diluted, and manipulated. By fully leveraging the laws of quantum mechanics, one could significantly reduce the size of a quantum dataset, which is infeasible in the classical scenario.

We now begin the more technical part of our paper by giving some background and defining the minimum training dataset problem. We note here that detailed mathematical proofs are given in the Appendix.

\section{Main results on unitary learning}
\subsection{The minimum training dataset problem} 
The strategy of learning an unknown unitary $U$ is to train a QML model $V$ with the training dataset $\cD=\{(\rho_j,U\rho_j U^\dagger)\}_{j=1}^t$ of size $t$ so that $V$ acts similarly to $U$ on the training set $\cD$. The goal is that the QML model $V$ can emulate the action of the unitary $U$ on any quantum state. The scheme is illustrated in Fig.~\ref{fig:work_flow}. \update{Notice that one needs to prepare many copies of each state in the dataset and repeatedly access them in the process of unitary learning.} Throughout this paper, we refer to input states $\rho_j$ as the training data, and target output states $U\rho_j U^\dagger$ are the corresponding labels.

\begin{figure}[t]
    \centering
    \includegraphics[width=\linewidth]{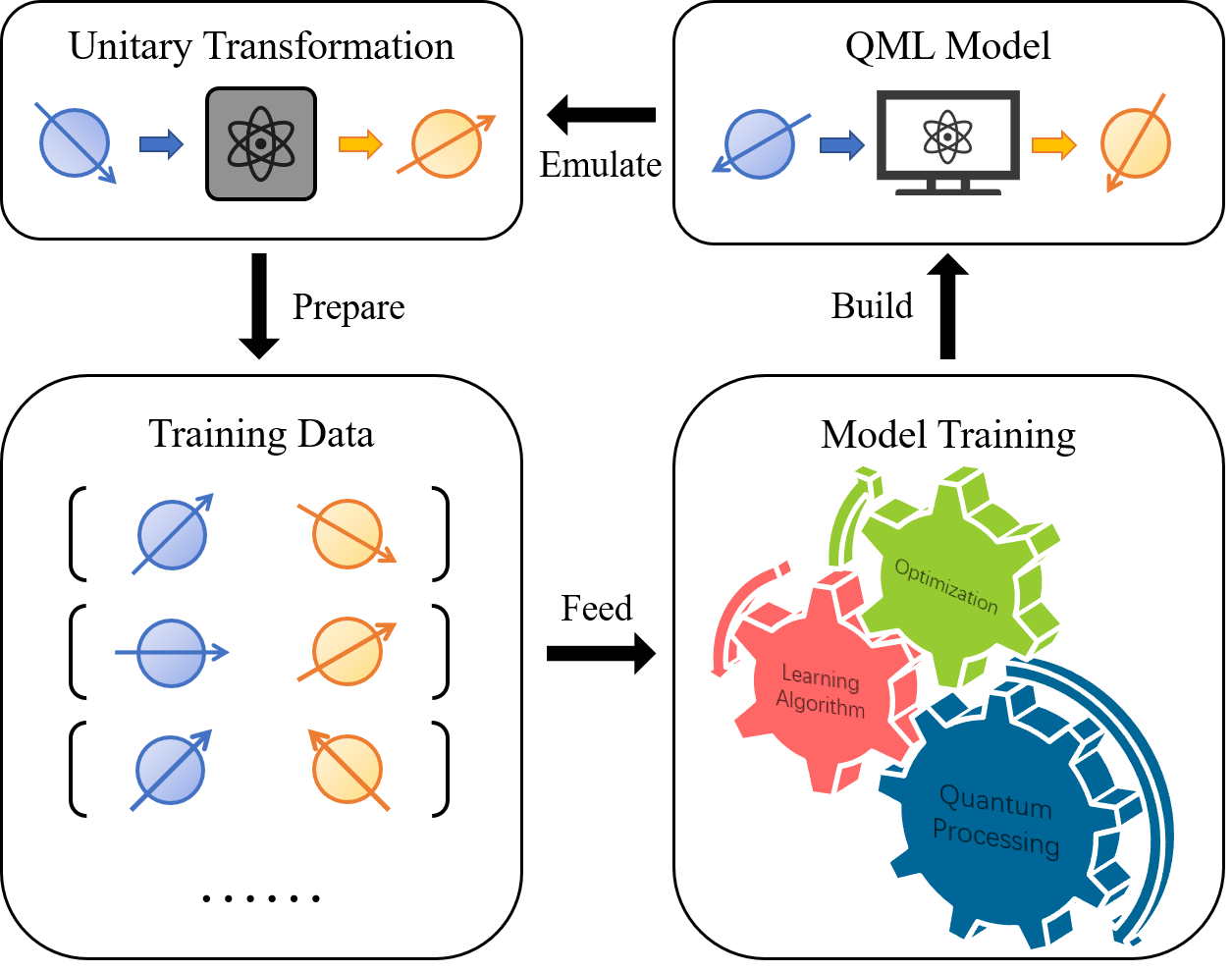}
    \caption{The main scheme of using QML to learn an unknown unitary transformation. For the target unitary transformation, the first step is to prepare the training dataset consisting of input and output quantum states. Next, we train a QML model with the data using quantum processing, machine learning algorithms, and optimization methods. After the training process, the built QML model is able to emulate the action of the target unitary transformation.}
    \label{fig:work_flow}
\end{figure}

To effectively learn unitary operators, one must choose a suitable loss function for performing optimization and quantifying how well a model has learned the target unitary. \update{Note that a standard method of directly evaluating the gate fidelity between the target unitary $U$ and the QML model $V$ replies on quantum
process tomography~\cite{Nielsen2010}, which is infeasible due to the exponential growth of required resource. Instead of using the gate fidelity, we choose the loss function that compares how well the output states of $U$ and $V$ are matched, which is arguably easier to evaluate on a quantum computer.} A common loss function is the squared trace distance between $U\rho_j U^\dagger$ and $V\rho_j V^\dagger$, i.e.
\begin{equation}
    \ell(U\rho_j U^\dagger, V\rho_j V^\dagger) = \norm{U\rho_j U^\dagger - V\rho_j V^\dagger}_1^2,
\end{equation}
where $\norm{\cdot}_1$ indicates the trace norm. \update{The trace distance could be evaluated by a variational quantum algorithm~\cite{chen2021variational}. There are many other distance measures that could also serve as the loss function, e.g.\ the Hilbert-Schmidt distance (or Frobenius distance), which could be evaluated by the \textsc{swap} test~\cite{buhrman2001quantuma}.} The \emph{quantum empirical risk} is defined by averaging the loss function on the training dataset,
\begin{equation}\label{eq:empirical risk}
    \hat{R}_U(V) = \frac{1}{t} \sum_{j=1}^t \ell(U\rho_j U^\dagger, V\rho_j V^\dagger).
\end{equation}

We quantify how well the model $V$ performs in simulating the unitary $U$ by the \emph{quantum risk}, defined as the squared trace distance between the outputs of $U$ and $V$ applied to the same input, averaged over all pure states induced by the uniform Haar measure,
\begin{align}
    R_U(V) &= \int \ell(U\ketbra{\psi}{\psi} U^\dagger, V\ketbra{\psi}{\psi} V^\dagger) d\ket{\psi}\\
    &= \int d\ket{\psi} \norm{U\ketbra{\psi}{\psi} U^\dagger - V\ketbra{\psi}{\psi}V^\dagger}_1^2\\
    &= 1 - \int d\ket{\psi} \abs{\bra{\psi}U^\dagger V\ket{\psi}}^2\\
    &= 1 - \frac{N + \abs{\tr[U^\dagger V]}^2}{N(N + 1)},
\end{align}
where $N = 2^n$ is the dimension of the corresponding Hilbert space.

Since our study focuses on the sufficiency of training datasets, for simplicity we assume the training is always \emph{perfect}, i.e. the quantum empirical risk $\hat{R}_U(V) = 0$. Under the assumption of perfect training, a training dataset is sufficient for the model $V$ to learn the unitary $U$ if the quantum risk $R_U(V) = 0$, otherwise it is insufficient.

The minimum training dataset problem for unitary learning is to find the minimum size $t_{\min}$ such that there exists a sufficient training dataset of size $t_{\min}$ to learn an unknown unitary $U$, and any dataset of size less than $t_{\min}$ is not sufficient.

\subsection{Minimum size of training datasets with pure states} The lower bound of expected quantum risk when applying datasets with pure states is given in Ref.~\cite{Poland2020}, which shows that it is not sufficient to learn an $n$-qubit unitary using less than $2^n$ pairs of Haar random pure states as the training data. For example, it could be inferred that the expected quantum risk of training $V$ with $O(n)$ pairs of pure states is $\EE_U[\EE_\cD[R_U(\theta)]] = 1 - O(\frac{1}{2^n})$.

However, the lower bound of the expected quantum risk does not imply $t_{min} = 2^n$ for training datasets with pure states. Consider using $2^n$ pairs of pure states as the training data. A natural way is to pick an orthonormal basis. Surprisingly, we find that using $2^n$ orthogonal states as the training data is not sufficient for learning a unitary, as the worst case of quantum risk is $R_U(V) = 1 - \frac{1}{2^n+1}$. For more details, we refer to \Cref{lemma: worst case for an orthogonal basis} in the Appendix.

Sharma et al.~\cite{Sharma2020a} give a discussion on the case that pure states in the training dataset are nonorthogonal but linearly independent, and conclude that this case leads to the same lower bound of average quantum risk as in Ref.~\cite{Poland2020}. Here we prove an upper bound of quantum risk when pure states in the training dataset are nonorthogonal but linearly independent.

\begin{restatable}[Upper bound of quantum risk]{proposition}{upperboundofquantumrisk}\label{prop:Upper bound of quantum risk}
Consider learning an $n$-qubit $U$ using the QML model $V$ with a training set $\cD=\{(\ket{x_j}, U\ket{x_j})\}_{j=1}^t$, where $\{\ket{x_j}\}_{j=1}^t$ is a set of nonorthogonal but linearly independent vectors. Assume that $\frac{N}{2} \leq t \leq N$, then we have $R_U(V) \leq 1 - \frac{N + (2t-N)^2}{N(N+1)}$.
\end{restatable}

The proof of this upper bound is based on the fact that $U^\dagger V = \cI$ in the $t$-dimensional subspace $\cH_\cD$ spanned by $\{\ket{x_j}\}_{j=1}^t$. In the other $(N-t)$-dimensional subspace $\cH_\cD^\perp$, $U^\dagger V$ could be any unitary matrix. Considering the worst case, $U^\dagger V = -\cI$ in $\cH_\cD^\perp$, then $\tr[U^\dagger V] = t - (N-t) = 2t - N$.

The upper bound of quantum risk implies that using $t = N$ nonorthogonal but linear-independent pure states is sufficient for learning a unitary. While $t < N$, it is not sufficient to learn a unitary since the worst case of quantum risk is greater than $0$. From the upper bound of quantum risk, we conclude that $t_{\min} = 2^n$ for training datasets with pure states.

The empirical risk is computed at each training iteration under the empirical risk minimization principle, which means the trace distance between $U\rho_j U^\dagger$ and $V\rho_j V^\dagger$ is evaluated for each state $\rho_j$ in the training set. When using $2^n$ pure states as the training data to learn a unitary, the overhead of this exponential scaling limits the efficiency of the learning process, and thus places a cutoff on the size of the target unitary in practical scenarios. 

\subsection{Efficient training dataset with mixed states} 
As we show in the previous section, using pure states as the training dataset to learn a unitary is not efficient because of the overhead exponentially scaling with the number of qubits. However, quantum data are not limited to pure states but also involve mixed states. In practice, quantum systems are open and interact with environments. Such systems are in mixed states as we have only incomplete information about the systems. Hence, it is natural and sensible to consider the training dataset with mixed states.

A mixed state $\rho$ is a mixture of pure states, i.e.,
\begin{equation}\label{eq:mixed_state}
    \rho\equiv \sum_j p_j \ketbra{\psi_j}{\psi_j},
\end{equation}
where each $\ket{\psi_j}$ is a pure state and each $p_j$ is a positive real number such that $\sum_j p_j = 1$. The rank of a mixed state $\rho$ is the rank of its density matrix in Eq.~\eqref{eq:mixed_state}. For an $n$-qubit mixed state, the maximum rank is $2^n$, which is also called full rank, and the minimum rank is $1$, in which case it is a pure state. Obviously, if a mixed state has a larger rank, it contains more information. Intuitively, fewer quantum states are needed to learn an unknown unitary when using mixed states as data.

To learn the target unitary $U$, we construct a training dataset $\cD_1$ that contains $n+1$ quantum states. Define
\begin{equation}\label{eq:rho0}
    \rho_0 = (\ketbra{+}{+})^{\otimes n}
\end{equation}
and 
\begin{equation}\label{eq:rhoj}
    \rho_j = \tau^{\otimes (j-1)} \otimes \ketbra{0}{0} \otimes \tau^{\otimes (n-j)}
\end{equation}
for $1 \leq j \leq n$, where $\tau$ is the maximally mixed state. For an $n$-qubit unitary $U$, let $\cD_1 = \{(\rho_j, U\rho_j U^\dagger)\}_{j=0}^n$ be the training dataset.

To show that the training dataset $\cD_1$ is sufficient for the QML model $V$ to learn $U$, we first prove that training with the dataset $\{(\rho_i, U\rho_i U^\dagger)\}_{j=1}^n$ makes $U^\dagger V$ a diagonal matrix. This is because training with each state decouples the space spanned by half of the computational basis states from the other half. Training with all these states together decouples the space spanned by each computational basis state from each other, which makes $U^\dagger V$ diagonal in the computational basis.
Then training with $(\rho_0, U\rho_0 U^\dagger)$ unifies all diagonal entries in $U^\dagger V$, which ensures that $U = V$.

\begin{restatable}{proposition}{linearsizeoftrainingset}\label{lemma:linear size of training dataset}
For an $n$-qubit unitary $U$ and a QML model $V$, if it satisfies that
\begin{equation}\label{eq:linear_perfect_train}
    U \rho_i U^\dagger = V \rho_i V^\dagger
\end{equation}
for $0 \leq i \leq n$, then we have $U = V$.
\end{restatable}

\Cref{lemma:linear size of training dataset} shows that we can learn a unitary using a QML model based on the training dataset $\cD_1$, which consists of $n+1$ pairs of quantum states. The size of the training dataset $\cD_1$ grows linearly with the number of qubits, which is more efficient than the case of using pure states. Furthermore, each state in $\cD_1$ is in a form of tensor product without entanglement, which is arguably easy to construct.

\subsection{Minimum size of datasets with mixed states} Now that using mixed states could reduce the size of the training dataset, it is of interest to study the ultimate limits on the case of using mixed states to learn a unitary. In this section, we further reduce the size of training sets from linear to constant, which yields an optimal quantum dataset for unitary learning. Moreover, we study the more general cases where the data could be randomly generated instead of constructed.

We first consider learning a unitary $U$ using the QML model $V$ with one mixed state $\rho$ as the training data. Suppose $U\rho U^\dagger = V \rho V^\dagger$, i.e.\ $U^\dagger V$ and $\rho$ commute, then $U^\dagger V$ and $\rho$ share a common set of eigenvectors~\cite[Lemma~5F]{Strang2006}. Since the eigenvalues of $U^\dagger V$ are not fixed, the worst case of quantum risk is $R_U(V) = 1 - \frac{1}{2^n+1}$, we hence conclude that using one mixed state as the training data is not sufficient for learning a unitary.

We then prove that using two mixed states as the training data is sufficient to learn an $n$-qubit unitary.

\begin{restatable}{proposition}{twomixedstatetoidentity}\label{prop:two mixed state to identity}
Suppose a unitary $W$ satisfies $W \rho W^\dagger = \rho$ and $W \sigma W^\dagger = \sigma$, where $\rho$ and $\sigma$ are full-ranked mixed states that have nondegenerate eigenvalues, and eigenvectors of $\rho$ and $\sigma$ are nonorthogonal. Then $W = \cI$.
\end{restatable}

In order to prove \Cref{prop:two mixed state to identity}, we show that $U^\dagger V$ is simultaneously diagonal with respect to the eigenbasis of $\rho$ and $\sigma$ if $U^\dagger V$ commute with the state $\rho$ and $\sigma$, respectively. Since eigenvectors of $\rho$ and $\sigma$ are mutually nonorthogonal, it implies that $U^\dagger V$ must be identity, i.e.\ $U = V$.

More generally, any two random full-ranked mixed states $\rho_r$ and $\sigma_r$ that are Hilbert-Schmidt distributed satisfy the condition in \Cref{prop:two mixed state to identity}, thus a random dataset $\cD_2 = \{(\rho_r, U\rho_r U^\dagger), (\sigma_r, U\sigma_r U^\dagger)\}$ is sufficient for learning an $n$-qubit unitary.


Our analytical results indicate that at least two quantum states are required for training a QML model to learn an unknown unitary. Thus we can conclude that $t_{\min} = 2$ for the quantum training datasets, which solves the minimum training dataset problem for unitary learning.

\textbf{Remark 1.}
\update{For the task of characterizing a quantum gate, which consists of distinguishing two unitaries as well as assessing the unitarity of a time evolution, Ref.~\cite{Reich2013} showed that the minimum number of input states required is two. On the one hand}, we note that perfectly distinguishing  any two unitaries reduces to the same mathematical problem as unitary learning: find a set of quantum states $\{\rho_j\}$ such that a unitary $U$ satisfies $U\rho_j U^\dagger = \rho_j$ for each $\rho_j$ if and only if $U$ is the identity matrix up to a global phase. Hence, utilizing the main tools and results in Ref.~\cite{Reich2013}, one \update{can conclude that two input states are sufficient for unitary learning.}

\update{On the other hand, since we assume that the target time evolution is unitary, the task of learning a unitary does not involve unitarity assessing. Thus, our results imply that two states are also necessary for distinguishing a pair of unitaries even if we do not need to assess the unitarity of time evolution.}
Moreover, by observing the equivalence between unitary learning and unitary distinguishing, \Cref{prop:Upper bound of quantum risk} implies that the minimal set of pure states sufficient for distinguishing any two unitaries has a size of $d$\update{, while the minimum number of pure states required for gate characterization is $d+1$ as shown in Ref.~\cite{Reich2013}.}


\textbf{Remark 2.}
\update{Notice that numerically generating density matrices in $\cD_2$ that satisfy the conditions in \Cref{prop:two mixed state to identity} is simple, following the method in Ref.~\cite{zyczkowski2001induced}. However, preparing such quantum states (also known as totally mixed thermal states) on a quantum computer is very difficult. As a comparison, the dataset $\cD_1$ contains $n+1$ elementary tensor product states as proposed in Eqs.~\eqref{eq:rho0} and~\eqref{eq:rhoj} that are easy to prepare. We remark that the dataset $\cD_1$ combines the properties of efficiency and feasibility, making it a particularly good choice for the task of unitary learning in practice.}

We next demonstrate applications of unitary learning to investigate the performance of optimal quantum datasets with two mixed states. \update{We adopt variational quantum algorithms in the numerical experiments, which use a classical optimizer to train a parameterized quantum circuit (PQC) as shown in Fig.~\ref{fig:PQC}.} Note that there are many reasonable choices of PQCs that differ from ours, such as alternating layered ansatzes~\cite{Nakaji2021} and variable structure ansatzes~\cite{Bilkis2021}, and the optimal choice may depend on the specifics of learning problem. \update{Our numerical experiments were carried out with the Paddle Quantum toolkit~\cite{Paddlequantum} on the PaddlePaddle Deep Learning Platform~\cite{Ma2019}, using a desktop with an 8-core i7 CPU and 32GB RAM.}

\begin{figure}
    \centering
    \includegraphics[width=0.32\textwidth]{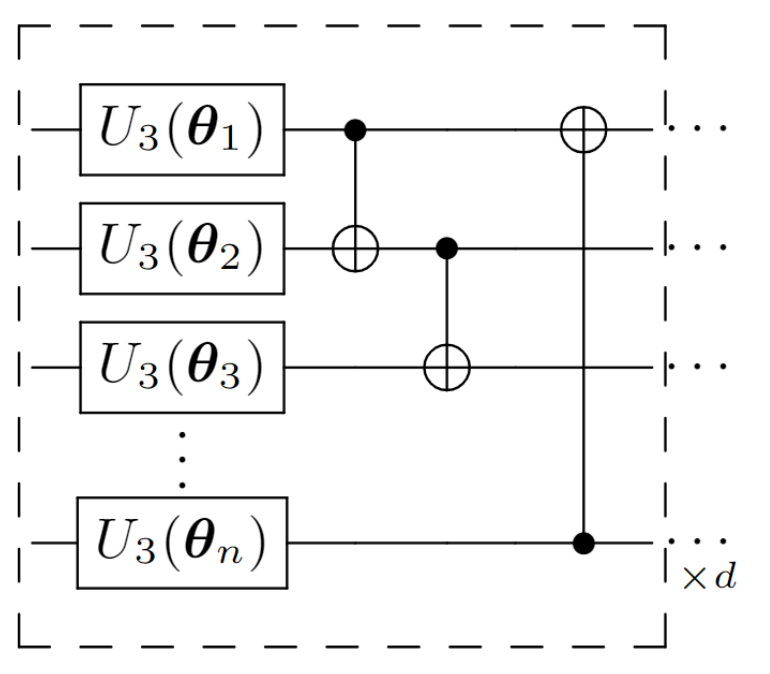}
    \caption{\reupdate{The parameterized quantum circuit, one layer of which has $n$ generic one-qubit rotation gates with three Euler angles $\bm\theta=(\theta, \phi, \lambda)$ and $n$ \textsc{cnot} gates. The \textsc{cnot} gates in the circuit are only between neighboring qubits on a 1D chain with periodic boundary condition.}}
    \label{fig:PQC}
\end{figure}

\section{Applications}
\subsection{Applications to Hamiltonian simulation}
One of the first proposed and the most natural applications of quantum computers is simulating other quantum systems~\cite{feynman1982simulating}. The dynamics of a closed quantum system are determined by a Hamiltonian $H$ (a Hermitian matrix), and the state of the system at time $t$ is represented by a state vector $\ket {\psi(t)}$. The system evolves over time according to the Schr\"{o}dinger equation:
 \begin{align}\label{eq:schrodinger}
 	|\psi(t)\rangle= e^{-i H t}|\psi(0)\rangle.
 \end{align}
The goal of quantum simulation is to produce the final state $\ket {\psi(t)}$ within some error tolerance for a given Hamiltonian $H$, an evolution time $t$, and an initial state $\ket {\psi(0)}$.


\begin{figure*}
\centering
\hfill
\captionsetup[subfigure]
{skip=-106pt,slc=off,margin={-5pt, 0pt}}
\subcaptionbox{\label{gatecount}}
{\includegraphics[width=0.32\textwidth]{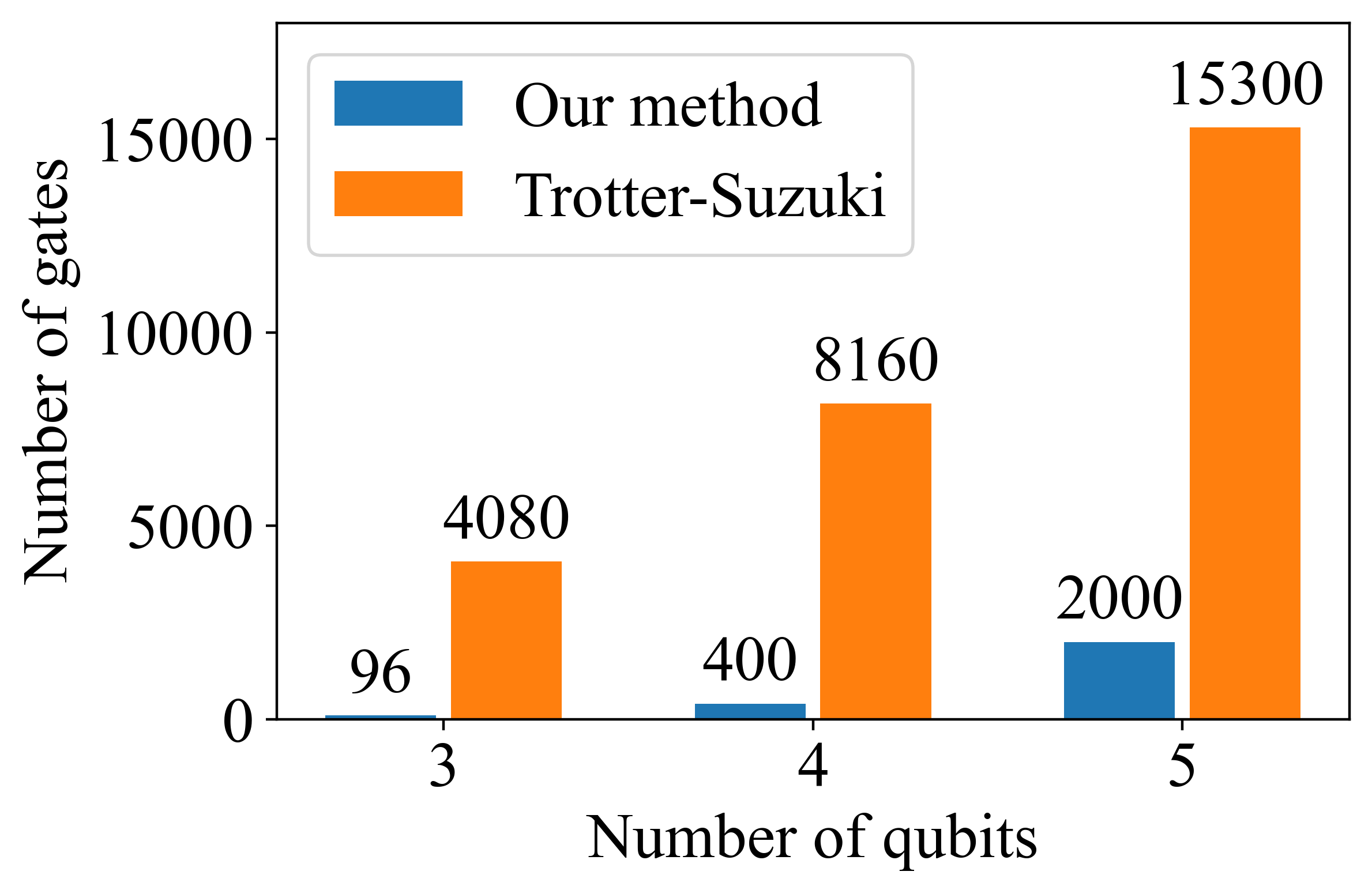}}
\hfill
\captionsetup[subfigure]
{skip=-106pt,slc=off,margin={-5pt, 0pt}}
\subcaptionbox{\label{d1fig}}
{\includegraphics[width=0.31\textwidth]{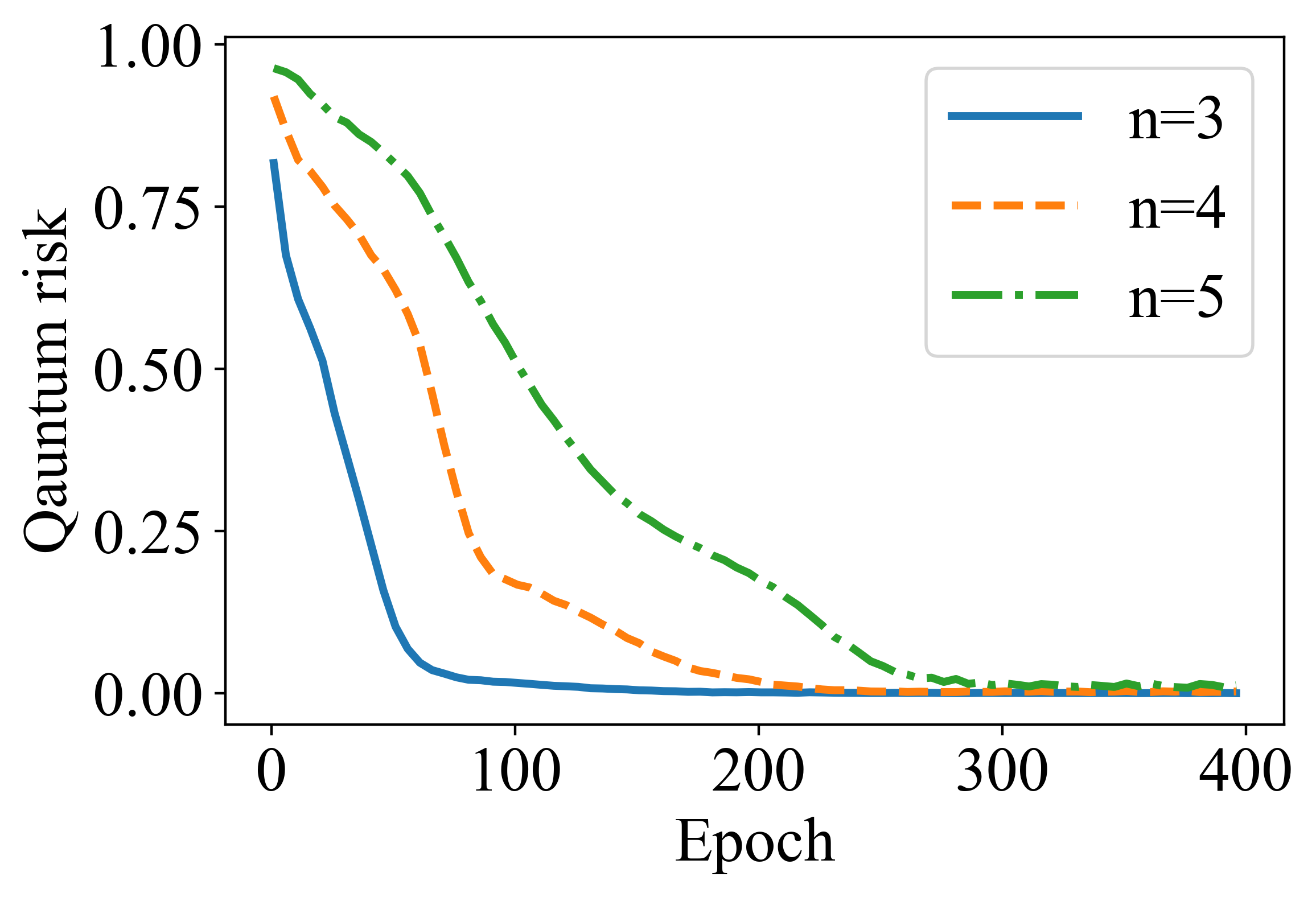}}
\hfill
\captionsetup[subfigure]
{skip=-106pt,slc=off,margin={-5pt,0pt}}
\subcaptionbox{\label{d2fig}}
{\includegraphics[width=0.31\textwidth]{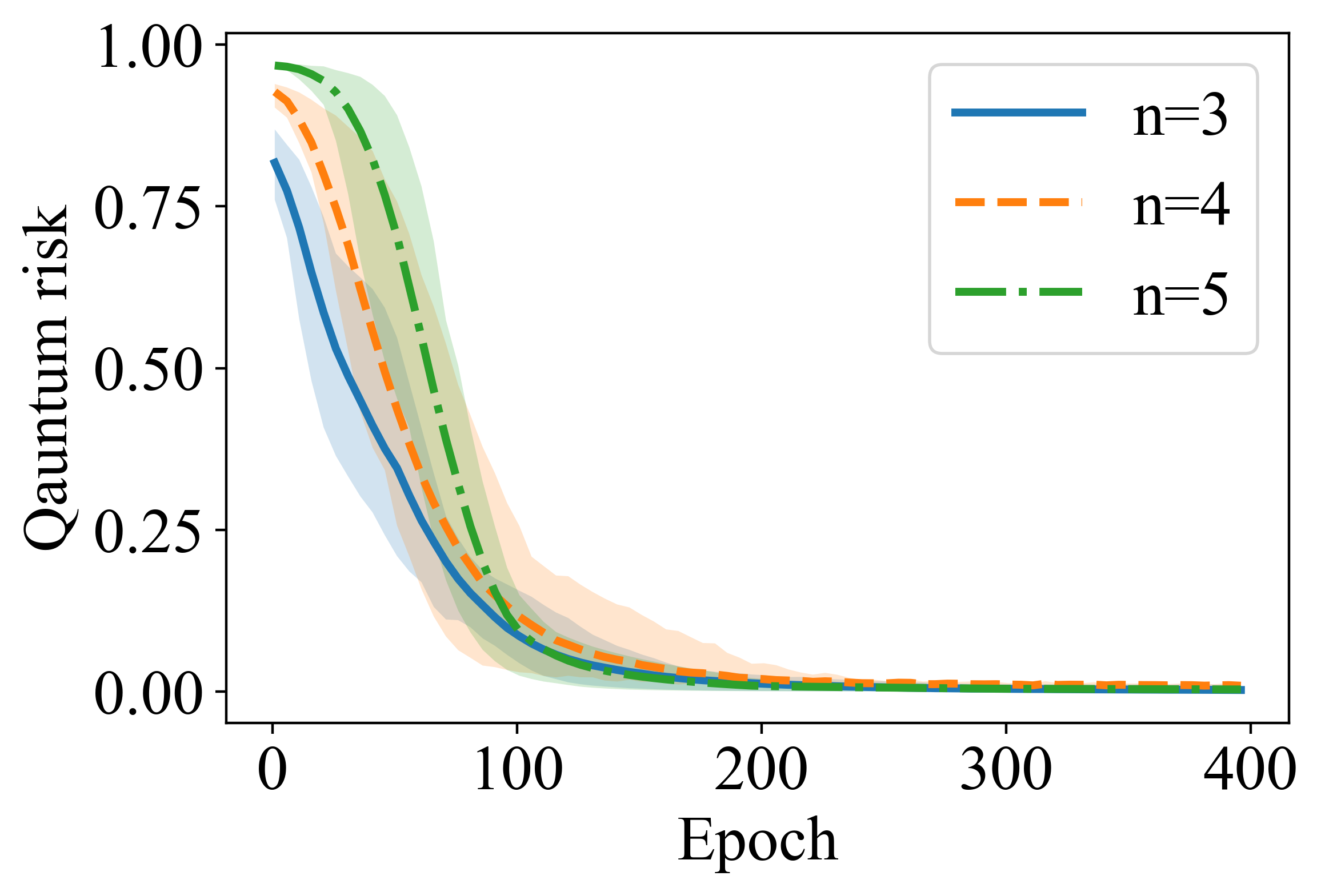}}
\captionsetup{subrefformat=parens}
\caption{Panel~\subref{gatecount} displays the number of required gates to simulate Hamiltonian of different quantum system sizes by our method and Trotter-Suzuki method. \reupdate{Note that each generic $U_3$ gates in the PQC is counted as three single-qubit rotation gates}. \update{Panel~\subref{d1fig} shows the quantum risk in training with the dataset $\cD_1$. The blue, orange and green lines correspond to the results on the quantum system with size $n=3, 4$ and $5$, respectively. Panel~\subref{d2fig} shows the quantum risk in training with randomly generated dataset $\cD_2$. The blue, orange and green lines are averaged over ten independent training instances for the three system sizes. The corresponding shaded region is the min-max threshold in different runs.}}
\label{fig:ham_sim_learn_pf}
\end{figure*}

The time evolution of a quantum system shown in Eq.~\eqref{eq:schrodinger} corresponds to a unitary operator $U = e^{-iHt}$. Simulating such time evolution on a quantum device requires the implementation of this unitary operator. If we can have access to the quantum system, it is possible to train a QML model to simulate $U$ with a training dataset sampled from the quantum process. On the other hand, if the physical evolution $U$ is not directly accessible but has been implemented as a quantum circuit $V_0$ through methods like product formula~\cite{suzuki1991general}, we can query $V_0$ to prepare training data and learn $V_0$ using another quantum circuit with a greatly reduced number of gates.

Here we consider learning a one-dimensional nearest-neighbor Heisenberg model~\cite{Childs2018a} with $n$ qubits, which is described by the Hamiltonian
\begin{equation}
    H = \sum_{j=1}^n (\Vec{\sigma}_j \cdot \Vec{\sigma}_{j+1} + h_j \sigma_j^z),
\end{equation}
where $\Vec{\sigma}_j = (\sigma_j^x, \sigma_j^y, \sigma_j^z)$ denotes a vector of Pauli $x$, $y$ and $z$ matrices on qubit $j$. Here a periodic boundary condition is imposed, i.e.\ $\Vec{\sigma}_{n+1} = \Vec{\sigma}_1$. The coefficient $h_j$ is chosen uniformly at random in $[-h, h]$, where the parameter $h$ characterizes the strength of the disorder.

\reupdate{In numerical experiments, we first generate an $n$-qubit one-dimensional nearest-neighbor Heisenberg model with random $h_j \in [0, 1]$. We then construct the time evolution circuit $V_0$ using the second-order Suzuki product formula~\cite{suzuki1991general} with evolution time $t=n$ and gate fidelity at least $1-10^{-3}$.} We train a PQC $V(\theta)$ with datasets $\cD_1$ and $\cD_2$, respectively, to learn $V_0$. The two mixed states in $\cD_2$ are randomly generated according to the Hilbert-Schmidt measure~\cite{zyczkowski2001induced}. Note that for the dataset $\cD_2$, we independently generate ten samples of datasets and average the experimental results over ten independent training instances in order to reduce the effect of randomness. The number of layers of PQC is chosen empirically, depending on the size of the quantum system. \reupdate{We select the number of layers $d = 8, 25,$ and $100$ for $n=3, 4,$ and $5$, respectively.} All parameters of trainable gates in PQCs are initialized following the uniform distribution on $[0, 2\pi]$. We choose the trace distance as the loss function and the empirical risk defined in Eq.~\eqref{eq:empirical risk} serves as the cost function in the training process. We set the training epoch to be 400 with a batch size equal to the size of datasets for all experiments. The Adam optimizer is used to minimize the cost function. \reupdate{The learning rate is $0.05$ and $0.1$ for training with dataset $\cD_1$ and $\cD_2$, respectively. For $n=5$, the learning rate decays by $0.5$ every $100$ epochs in order to achieve a better convergence}. 

\update{As the numerical results shown in Fig.~\ref{fig:ham_sim_learn_pf}, PQC models trained with the dataset $\cD_1$ and randomly generated datasets $\cD_2$ could simulate the Hamiltonian with high accuracy (the average quantum risk achieves $0.01$). More importantly, the number of gates in PQC models is significantly fewer compared to circuits constructed by the second-order Suzuki product formula.}

\subsection{Applications to oracle compiling}
A \emph{quantum oracle} is a ``black-box'' operator that is used extensively in quantum algorithms. Most known quantum algorithms, such as the Shor's algorithm~\cite{Shor1997}, Grover's algorithm~\cite{Grover1996}, quantum walks~\cite{Szegedy2004}, and quantum singular value transformation~\cite{Gilyen2019}, are constructed using quantum oracles. 

A quantum oracle is usually defined using a classical Boolean function $f:\{0,1\}^n \to \{0,1\}$, which maps an $n$-bit binary input to a binary output. The most common definition of quantum oracles is
\begin{equation}
    U_f ( \ket{x} \otimes \ket{y} ) = \ket{x}\otimes\ket{y \oplus f(x)},
\end{equation}
where ``$\oplus$'' is the \textsc{xor} operation. Notice that the oracle $U_f$ is unitary by the definition. In practice, such binary oracles are often used in quantum algorithms to do a phase rotation. This means that it is simpler to consider an alternative oracle, which is called \emph{phase oracle}, defined as follows:
\begin{equation}
    U_p \ket{x}= (-1)^{f(x)}\ket{x}.
\end{equation}

As a ``black-box'' operator, the structure of a quantum oracle is typically unknown, thus we may not able to directly implement it on a quantum device. Based on the idea of unitary learning, we could train a QML model to learn an unknown quantum oracle, which permits us to compile and run an oracle-based quantum algorithm on quantum devices.

Here we showcase an example of compiling a phase oracle and then run the famous Grover's algorithm based on this oracle. Suppose the search space consists of $N = 2^n$ elements which we label with $n$-bit strings. Let $f:\{0, 1\}^n \to \{0, 1\}$ be the Boolean function telling which elements are marked: the element labeled with $x\in\{0,1\}^n$ is marked if $f(x) = 1$ and unmarked otherwise. Let $U_p$ be the phase oracle that encodes the Boolean function $f$, and we assume there is a unique marked element. Let $D = 2 \ketbra{\Psi}{\Psi} - \cI$ be the Grover diffusion operator, where $\ket{\Psi} = \ket{+}^{\otimes n}$. The operator $G = D\cdot U_p$ is known as the Grover iterate. Grover's algorithm repeatedly applies the Grover iterate $G$ for $\lceil \frac{4}{\pi}\sqrt{N} \rceil$ times on the initial state $\ket{\Psi}$ and measures the resulting state. The success probability of Grover's algorithm is the probability of obtaining the marked element when making the measurement.

\update{In numerical experiments, we train a PQC $V(\theta)$ with dataset $\cD_1$ and random datasets $\cD_2$, respectively, to learn the phase oracle $U_p$. Except for different target unitaries, the setting of training PQCs is the same as what was used in the previous application. Then we append the circuit of Grover diffusion operator $D$ to $V(\theta)$ as the Grover iterate. By repeatedly applying the entire circuit of Grover iterate, we are able to run the compiled version of Grover's algorithm. The experimental results with a comparison to the theoretical success probability of Grover's algorithm are listed in Table~\ref{table:compile grover}. From the results, we could see that, by training a PQC with our proposed datasets $\cD_1$ and $\cD_2$, our method compiles the oracle into a sequence of single-qubit and two-qubit gates with high accuracy, which enables us to run an oracle-based algorithm on near-term quantum devices.}

\begin{table}[htbp]
\centering
\begin{tabular}{ccccc}
\toprule
\setlength{\tabcolsep}{5em}
\multirow{2}{*}[-2pt]{No.\ of qubits} & \multicolumn{3}{c}{Success probability}\\
\cmidrule{2-4}
& Grover & Compiled with $\cD_1$ &  Compiled with $\cD_2$\\
\midrule
3 & $0.94531250$  & $0.94446218$  &  $0.94862192$  \\
\addlinespace
4 & $0.96131897$ & $0.96496582$  & $0.95465205$ \\
\addlinespace
5 & $0.99918232$  & $0.99569303$  & $0.99097359$ \\
\bottomrule
\end{tabular}
\caption{\update{Performance of compiling Grover's algorithm with datasets $\cD_1$ and $\cD_2$. The success probability of compiled Grover's algorithm with dataset $\cD_2$ is averaged over ten independent training instances.}}
\label{table:compile grover}
\end{table}


\section{Concluding remarks} Quantum computers could be used to simulate the dynamics of quantum systems, which is known to be computationally difficult for classical computers. One promising method is to train a quantum machine model based on quantum data, so that the model is able to mimic a unitary transformation. The performance of any machine learning model highly relies on the training data. Typical methods for learning unitaries often employ exponentially large training datasets~\cite{Cincio2018, Cincio2021, Beer2020, Poland2020}. In this work, we identify mixed states as the key to reduce the size of quantum datasets that are sufficient for learning an unknown unitary operator. We introduce efficient quantum datasets with mixed states that are exponentially smaller than datasets consisting of pure states. Explicitly, we show that the optimal size of quantum datasets is $t_{\min} = 2$ and we also obtain a practical dataset consisting of elementary tensor product states. Our analytical results are completely general and thus are not restricted to some particular quantum machine learning models. \update{We showcase the applications of our results in Hamiltonian simulation and oracle compiling under the framework of hybrid quantum-classical algorithms, which demonstrates the effectiveness and practicability of our results in meaningful tasks.}

Our results may also shed light on the advantages of quantum machine learning over classical machine learning, with respect to the training data. The main goal of using quantum machine learning to simulate quantum dynamics is to provide speedups over classical methods, as the computational cost of quantum simulation using classical computers is believed to grow exponentially with system size. Thus any exponential scaling in the quantum machine learning algorithm places a barrier on this speedup. Our results reduce the size of the dataset from exponentially large to linear, and further to a constant, which removes such a barrier.

As the minimum training dataset problem for unitary learning has been answered, one future step is to study the more general case. Note that the most general quantum process is a quantum channel, it is of great interest to study the sufficient training dataset for learning an unknown quantum channel. \reupdate{Possible applications of quantum channel learning include quantum error mitigation~\cite{Endo2020,Temme2017,Strikis2020,Jiang2020,Piveteau2021,Cai2022} and quantum channel simulation~\cite{Fang2018,Caruso2014,Wilde2018,Diaz2018}.} The training dataset problem for learning a quantum channel may also have implications for the learnability of QML models. Another interesting direction is to study the relationship between the size of the quantum dataset and the generalization error of unitary learning, without the assumption of perfect training. \update{Our work shows applications of unitary learning via a quantum dataset, it will also be interesting to discover more such applications in quantum information processing and quantum machine learning~\cite{Biamonte2017a,Schuld2018b,Arunachalam2017}.}

\textbf{Acknowledgements.}
We would like to thank Runyao Duan and Yin Mo for their helpful discussions. 
We also thank Christiane Koch for pointing out the connections between differentiating unitaries and learning unitaries.
This work was done when Z. Y., X. Z., and B. Z. were research interns at Baidu Research.

 
\bibliography{baidu2}

\vspace{1cm}
\onecolumngrid
\vspace{1cm}
\begin{center}
{\textbf{\large Appendix}}
\end{center}

\renewcommand{\theequation}{S\arabic{equation}}
\renewcommand{\theproposition}{S\arabic{proposition}}
\renewcommand{\thelemma}{S\arabic{lemma}}
\renewcommand{\thecorollary}{S\arabic{corollary}}
\renewcommand{\thefigure}{S\arabic{figure}}
\setcounter{equation}{0}
\setcounter{table}{0}
\setcounter{section}{0}
\setcounter{proposition}{0}
\setcounter{lemma}{0}
\setcounter{corollary}{0}
\setcounter{figure}{0}

\section{Proofs of the main results}
\begin{lemma}\label{lemma: worst case for an orthogonal basis}
Consider learning an $n$-qubit unitary $U$ using a QML model $V$ with training set $\cD=\{(\ket{x_j}, U\ket{x_j})\}_{j=1}^N$, where $\{\ket{x_j}\}_{j=1}^N$ is an orthonormal basis. Then the worst case of learning is $R_U(V) = 1 - \frac{1}{2^n+1}$.
\end{lemma}
\begin{proof}
Since the training is perfect, i.e.\ $V\ket{x_j} = e^{i\theta_j} U\ket{x_j}$ for $j = 1,2,\ldots, N$, we have $\bra{x_j}U^\dagger V\ket{x_j} = e^{i\theta_j}$. Then we can write
\begin{equation}
    U^\dagger V = \sum_{j=1}^N e^{i\theta_j}\ketbra{x_j}{x_j}.
\end{equation}
Note that each $\theta_j$ are not necessarily the same. Consider a worst case scenario, $\theta_j = \theta$ for $j=1,\ldots,\frac{N}{2}$ and $\theta_j = \theta + \pi$ for $j=\frac{N}{2}+1,\ldots,N$, where $\theta$ is an arbitrary angle, then we have $\tr[U^\dagger V]=0$ and $R_U(V) = 1 - \frac{1}{2^n+1}$.
\end{proof}

\upperboundofquantumrisk*
\begin{proof}
Training the QML model $V$ with the training set $\cD$ makes
\begin{equation}
U^\dagger V =
\left[\begin{array}{@{}c|c@{}}
  \begin{matrix}
  e^{i\theta_1} & \cdots & 0 \\
  \vdots & \ddots & \vdots\\
  0 & \cdots & e^{i\theta_t}
  \end{matrix}
  & \mbox{$0$} \\
\hline
  \mbox{$0$} &
  \mbox{$Y$}
\end{array}\right],
\end{equation}
where $Y$ is a unitary matrix on an $(N-t)$-dimensional Hilbert space. Since $\{\ket{x_j}\}_{j=1}^t$ is a set of nonorthogonal but linearly independent vectors, we have $\theta_j = \theta_k, \forall j,k \in \{1,\ldots, t\}$~\cite{Sharma2020a}. Then we can write
\begin{equation}
U^\dagger V =
\left[\begin{array}{@{}c|c@{}}
  \mbox{$e^{i\theta} \cI$}
  & \mbox{$0$} \\
\hline
  \mbox{$0$} &
  \mbox{$Y$}
\end{array}\right].
\end{equation}
Thus we have
\begin{align}
    \abs{\tr[U^\dagger V]} &= \abs{\tr[e^{i\theta}\cI] + \tr[Y]}\\
    &= \abs{\tr[e^{i\theta}\cI] - \tr[-Y]}\\
    &\geq \abs{\abs{\tr[e^{i\theta}\cI]} - \abs{\tr[-Y]}}\label{eq:reverse triangle inequality}\\
    &= \abs{t - \abs{\tr[Y]}}\\
    &\geq t - (N-t) \label{eq:maximum trace of Y}\\
    &= 2t - N.
\end{align}
Inequality~\eqref{eq:reverse triangle inequality} is from the reverse triangle inequality. Inequality~\eqref{eq:maximum trace of Y} follows from the fact that the maximum of $\tr[Y]$ is $(N-t)$. The upper bound of the quantum risk is
\begin{equation}
    R_U(V) = 1 - \frac{N+\abs{\tr[U^\dagger V]}^2}{N(N+1)} \leq 1 - \frac{N + (2t-N)^2}{N(N+1)}.
\end{equation}
\end{proof}

\begin{corollary}
Consider learning an $n$-qubit $U$ using the QML model $V$ with a training set $\cD=\{(\ket{x_j}, U\ket{x_j})\}_{j=1}^N$, where $\{\ket{x_j}\}_{j=1}^N$ is a set of nonorthogonal but linear independent vectors. Then we have $R_U(V) = 0$.
\end{corollary}
\begin{proof}
Directly follows from \Cref{prop:Upper bound of quantum risk}.
\end{proof}

\begin{lemma}\label{lemma:linear_mixed_diagonal}
A unitary operator $W$ satisfies
\begin{align}\label{eq:lemma1_same}
    W \rho_j W^\dagger = \rho_j
\end{align}
for all $1 \leq j \leq n$ only if $W$ is diagonal in the computational basis.
\end{lemma}
\begin{proof}
A state $\rho_j$ defined in Eq.~\eqref{eq:rhoj} can be written out in the computational basis as
\begin{align}
    \rho_j &= \frac{1}{2^{n-1}} \sum_{k=0}^{2^j-1} \sum_{l=0}^{2^{n-j-1}-1} \proj{l + k\cdot 2^{n-j}}.
\end{align}
The unitary operator $W$ can be written as $W = \sum_{k=0}^{2^{n}-1} e^{-i\theta_k} \ket{\phi_k}\bra{k}$ for an orthonormal set of states $\{\phi_k\}_{k=0}^{2^{n-1}}$. Hence we have
\begin{align}
    W \rho_j W^\dagger &= \left(\sum_{r=0}^{2^{n}-1} e^{-i\theta_r} \ket{\phi_r}\bra{r}\right) \rho_j \left(\sum_{s=0}^{2^{n}-1} e^{i\theta_s} \ket{s}\bra{\phi_s}\right)\\
    &= \frac{1}{2^{n-1}} \sum_{k=0}^{2^j-1} \sum_{l=0}^{2^{n-j-1}-1} \proj{\phi_{l + k\cdot 2^{n-j}}}\\
    &= \frac{1}{2^{n-1}} \sum_{k=0}^{2^j-1} \sum_{l=0}^{2^{n-j-1}-1} \proj{l + k\cdot 2^{n-j}}, \label{eq:lemma1_proof_same}
\end{align}
where Eq.~\eqref{eq:lemma1_proof_same} follows from Eq.~\eqref{eq:lemma1_same}.
Then, the orthonormal sets of vectors $\Phi_j \equiv \{\ket{\phi_{l + k\cdot 2^{n-j}}}\}$ and $\Lambda_j \equiv \{\ket{l + k\cdot 2^{n-j}}\}$ must span the same subspace, and so do $\overline{\Phi_j} \equiv \{\ket{\phi_m} \vert m=0,\dots,2^n-1\} - \Phi_j$ and $\overline{\Lambda_j} \equiv \{\ket{m} \vert m=0,\dots,2^n-1\} - \Lambda_j$. Hence, we can say that the space spanned by $\Phi_j$ is decoupled from the space spanned by $\overline{\Phi_j}$ in the sense that $W = U_{\Phi_j} + U_{\overline{\Phi_j}}$, where $U_{\Phi_j}$ acts nontrivially only on the space spanned by ${\Phi_j}$.

Writing $l + k\cdot 2^{n-j}$ into its binary representation, we note that the set $\Lambda_j$ corresponds to all the bit strings whose $(j+1)$-th bits (from left to right) are $0$. Since for every pair of different computational basis states, there is at least one-bit difference in their corresponding bit strings, there exists a number $j$ such that one of this pair is in $\Lambda_j$ and the other is in $\overline{\Lambda_j}$. Thus, Eq.~\eqref{eq:lemma1_same} holding for all $j$ implies that the space spanned by each computational basis state is decoupled from each other: that is,
\begin{align}
    W = \sum_{k=0}^{2^n-1} U_k = \sum_{k=0}^{2^n-1} e^{-i\theta_k} \proj{k},
\end{align}
where $U_k = e^{-i\theta_k} \proj{k}$. Therefore, $W$ is a diagonal matrix.
\end{proof}

\begin{lemma}\label{lemma:plus_state_makes_identity}
If a diagonal unitary $\Lambda$ satisfies
\begin{equation}
    \Lambda\rho_0 \Lambda^\dagger = \rho_0,
\end{equation}
then $\Lambda = e^{-i\theta} \cI$, where $\cI$ is the identity matrix.
\end{lemma}
\begin{proof}
The pure state $\ket{+} = \frac{1}{\sqrt{2}} (1,1)^T$, then we can write
\begin{align}
    \biggl(\sum_{j=0}^N e^{-i\theta_j}\ketbra{j}{j}\biggr) \cdot (\ketbra{+}{+})^{\ox n} \cdot  \biggl(\sum_{k=0}^N e^{i\theta_k}\ketbra{k}{k}\biggr)
    &= \biggl(\sum_{j=0}^N e^{-i\theta_j}\ketbra{j}{j} \cdot \ket{+}^{\ox n}\biggr) \biggl(\sum_{k=0}^N e^{i\theta_k} \bra{+}^{\ox n}\ketbra{k}{k}\biggr)\\
    &= (\ketbra{+}{+})^{\ox n}.
\end{align}
Since $\bra{j}(\ket{+}^{\ox n}) = (\bra{+}^{\ox n})\ket{k} = 1/\sqrt{N}$, we have
\begin{equation}
    \frac{1}{N}\sum_{j=0}^N \sum_{k=0}^N e^{-i\theta_j} e^{i\theta_k} \ket{j} \bra{k} = (\ketbra{+}{+})^{\ox n}.
\end{equation}
The state $\rho$ is the matrix with all elements are the same, i.e., $\rho=(\ketbra{+}{+})^{\ox n}=\frac{1}{N}\mathbb{I}$, the elements in $\mathbb{I}$ are all $1$. Thus we have
\begin{equation}
    \sum_j \sum_k e^{-i\theta_j} e^{i\theta_k} \ketbra{j}{k} = \mathbb{I},
\end{equation}
which implies each element on the left-hand side equals to 1, so the phase term must be 1, i.e. $\theta_j=\theta_k=\theta$. Then the diagonal matrix can be written into 
\begin{equation}
    \sum_j e^{-i\theta}\ketbra{j}{j}=e^{-i\theta}\sum_j \ketbra{j}{j} = e^{-i\theta} \mathcal{I}.
\end{equation}
\end{proof}

\linearsizeoftrainingset*
\begin{proof}
Writing Eq.~\eqref{eq:linear_perfect_train} as $U^\dagger V \rho_i V^\dagger U = \rho_i$. By \Cref{lemma:linear_mixed_diagonal,,lemma:plus_state_makes_identity}, we have $U^\dagger V = e^{-i\theta}\cI$.
\end{proof}

\begin{lemma}\label{lemma:one mixed state is not sufficient}
Consider learning an $n$-qubit $U$ using the QML model $V$ with only one training data $\cD=\{(\rho,U\rho U^\dagger)\}$, where $\rho$ is a mixed state. The worst case of learning is $R_U(V)=1-\frac{1}{2^n+1}$.  
\end{lemma}
\begin{proof}
We assume the training is perfect, i.e. $U\rho U^\dagger  = V\rho V^\dagger$. Then, we rewrite it into 
\begin{equation}
    U^\dagger V \rho V^\dagger U = \rho.
\end{equation}
Since $U^\dagger V$ commutes with $\rho$, $U^\dagger V$ and $\rho$ share the same eigenvectors~\cite[5F]{Strang2006}. Suppose $\rho$ has eigenvectors $\{\psi_j\}_j$, we can write
\begin{equation}
    U^\dagger V = \sum_{j=1}^N e^{i\theta_j}\ketbra{\psi_j}{\psi_j}.
\end{equation}
Then we have
\begin{equation}
    \abs{\tr[U^\dagger V]} = \abs{\sum_{i=1}^N e^{i\theta_j}}.
\end{equation}
Note that each $\theta_j$ are not necessarily the same. Consider a worst case scenario, $\theta_j = \theta$ for $j=1,\ldots,\frac{N}{2}$ and $\theta_j = \theta + \pi$ for $j=\frac{N}{2}+1,\ldots,N$, where $\theta$ is an arbitrary angle, then we have $\tr[U^\dagger V]=0$ and $R_U(V) = 1 - \frac{1}{2^n+1}$.
\end{proof}

\twomixedstatetoidentity*
\begin{proof}
Write mixed states $\rho$ and $\sigma$ in the form of spectral decomposition: 
\begin{align}
    \rho &= \sum_{j} b_j\ketbra{x_j}{x_j},\\
    \sigma &= \sum_{j} c_j\ketbra{y_j}{y_j},
\end{align}
where $\{\ket{x_j}\}_j$ and $\{\ket{y_j}\}_j$ are two different eigenbases. By Ref.~\cite[5F]{Strang2006}, $W$ shares a common eigenbasis with $\rho$ and $\sigma$, respectively. Suppose that $\{\t_j\}_{j=1}^{N}$ is the eigenphases of $W$, we can write $W$ as
\begin{align}
    W &= \sum_{j = 1}^N e^{i\t_j} \ketbra{x_j}{x_j},\\
    W &= \sum_{j = 1}^N e^{i\t_j} \ketbra{y_j}{y_j},
\end{align}
where $\{\ket{x_j}\}_{j=1}^N$ and $\{\ket{y_j}\}_{j=1}^N$ are two different eigenbases of that satisfy $\braket{x_j}{y_k} \neq 0$ for each $1 \leq j,k \leq N$. Since eigenvectors with distinct eigenvalues are orthogonal, $\ket{x_j}$ and $\ket{y_k}$ being nonorthogonal implies that $\a_j = \a_k$ for each $1 \leq j,k \leq N$. Thus we have $W = e^{i\theta} \cI$.
\end{proof}

\begin{corollary}
Consider learning an $n$-qubit $U$ using the QML model $V$ with a training set $\cD_2 = \{(\rho_r, U\rho_r U^\dagger), (\sigma_r, U\sigma_r U^\dagger)\}$, where $\rho_r$ and $\sigma_r$ are two random generated full-ranked mixed states according to the Hilbert-Schmidt measure. Then we have $\EE_U[\EE_\cD[R_U(V)]] = 0$.
\end{corollary}
\begin{proof}
The Hilbert-Schmidt measure is equivalent to the measures induced by partial tracing of bipartite composite systems on the Hilbert space $\cH_N \otimes \cH_K$ if $N = K$~\cite{zyczkowski2001induced}. One could generate an $N \times N$ random density matrix $\rho_r$ according to the Hilbert-Schmidt measure by taking normalized random Wishart matrices $AA^\dagger$, with $A$ belonging to the Ginibre ensemble of Hermitian matrices of appropriate dimension. The probability distribution in the simplex of eigenvalues of $\rho_r$ is
\begin{equation}\label{eq:prob dist of eigenvalues}
    P(\lambda_1, \lambda_2, \ldots, \lambda_N) = C_N\delta(1-\sum_{j=1}^N\lambda_j)\prod_{j<k}^N(\lambda_j - \lambda_k)^2,
\end{equation}
where $C_N$ is the normalization constant. From the probability distribution in Eq.~\eqref{eq:prob dist of eigenvalues}, we observe that the probability of $\rho_r$ having two identical eigenvalues is $0$. It implies that a random density matrix has nondegenerate eigenvalues. In addition, eigenvectors of two randomly generated mixed states are nonorthogonal, as the probability that two independent random vectors that follow an absolutely continuous distribution will be exactly orthogonal is $0$. Thus $\rho_r$ and $\sigma_r$ satisfy the condition in \Cref{prop:two mixed state to identity}, which implies that $\EE_U[\EE_\cD[R_U(V)]] = 0$.
\end{proof}
\end{document}